\documentclass[11pt,leqno]{amsart}
\usepackage{amssymb,amsmath,amsthm,amsfonts}
\usepackage{latexsym}
\usepackage{graphicx}
\usepackage{subfig}

\newcommand{\Hmm}[1]{\leavevmode{\marginpar{\tiny%
$\hbox to 0mm{\hspace*{-0.5mm}$\leftarrow$\hss}%
\vcenter{\vrule depth 0.1mm height 0.1mm width \the\marginparwidth}%
\hbox to 0mm{\hss$\rightarrow$\hspace*{-0.5mm}}$\\\relax\raggedright #1}}}
\newcommand{\nc}{\newcommand}
\nc{\les}{\lesssim}
\nc{\nit}{\noindent}
\nc{\nn}{\nonumber}
\nc{\D}{\partial}
\nc{\diff}[2]{\frac{d #1}{d #2}}
\nc{\diffn}[3]{\frac{d^{#3} #1}{d {#2}^{#3}}}
\nc{\pdiff}[2]{\frac{\partial #1}{\partial #2}}
\nc{\pdiffn}[3]{\frac{\partial^{#3} #1}{\partial{#2}^{#3}}}
\nc{\abs}[1] {\lvert #1 \rvert}
\nc{\cAc}{{\cal A}_c}
\nc{\cE}{{\cal E}}
\nc{\cF}{{\mathcal F}}
\nc{\cP}{{\cal P}}
\nc{\cV}{{\cal V}}
\nc{\cQ}{{\cal Q}}
\nc{\cGin}{{\cal G}_{\rm in}}
\nc{\cGout}{{\cal G}_{\rm out}}
\nc{\cO}{{\cal O}}
\nc{\Lav}{{\cal L}_{\rm av}}
\nc{\cL}{{\cal L}}
\nc{\cB}{{\cal B}}
\nc{\cZ}{{\cal Z}}
\nc{\cR}{{\cal R}}
\nc{\cT}{{\cal T}}
\nc{\cY}{{\cal Y}}
\nc{\cX}{{\cal X}}
\nc{\cXT}{{{\cal X}(T)}}
\nc{\cBT}{{{\cal B}(T)}}
\nc{\vD}{{\vec \mathcal{D}}}
\nc{\efield}{\mathcal{E}}
\nc{\vE}{{\vec \efield}}
\nc{\vB}{{\vec \mathcal{B}}}
\nc{\vH}{{\vec \mathcal{H}}}
\nc{\ty}{{\tilde y}}
\nc{\tu}{{\tilde u}}
\nc{\tV}{{\tilde V}}
\nc{\Pc}{{\bf P_c}}
\nc{\bx}{{\bf x}}
\nc{\bX}{{\bf X}}
\nc{\bXYZ}{{\bf XYZ}}
\nc{\bY}{{\bf Y}}
\nc{\bF}{{\bf F}}
\nc{\bS}{{\bf S}}
\nc{\dV}{{\delta V}}
\nc{\dE}{{\delta E}}
\nc{\TT}{{\Theta}}
\nc{\dPsi}{{\delta\Psi}}
\nc{\order}{{\cal O}}
\nc{\Rout}{R_{\rm out}}
\nc{\eplus}{e_+}
\nc{\eminus}{e_-}
\nc{\epm}{e_\pm}
\nc{\eps}{\varepsilon}
\nc{\vnabla}{{\vec\nabla}}
\nc{\G}{\Gamma}
\nc{\w}{\omega}
\nc{\mh}{h}
\nc{\mg}{g}
\nc{\vphi}{\varphi}
\nc{\tlambda}{\tilde\lambda}
\nc{\be}{\begin{equation}}
\nc{\ee}{\end{equation}}
\nc{\ba}{\begin{eqnarray}}
\nc{\ea}{\end{eqnarray}}

\nc{\g}{\gamma}
\nc{\ol}{\overline}

\newtheorem{theo}{Theorem}

\newtheorem{prop}{Proposition}

\def\R{{\rm \rlap{\rm I}\,\bf R}}

\nc{\pT}{\partial_T}
\nc{\pz}{\partial_z}
\nc{\pt}{\partial_t}
\nc{\la}{\langle}
\nc{\ra}{\rangle}
\nc{\infint}{\int_{-\infty}^{\infty}}
\nc{\halfwidth}{6.5cm}
\nc{\figwidth}{10cm}

\nc{\nlayers}{L} \nc{\nsectors}{M}
\nc{\indicator}{\mathbf{1}}
\nc{\Rhole}{R_{\rm hole}}
\nc{\Rring}{R_{\rm ring}}
\nc{\neff}{n_{\rm eff}}
\nc{\Frem}{F_{\rm rem}}
\nc{\Real}{\mathbb R}
\nc{\Z}{\mathbb Z}
\nc{\DD}{\Delta}
\nc{\cD}{\mathcal D}
\nc{\lnorm}{\left\|}
\nc{\rnorm}{\right\|}
\nc{\rnormp}{\right\|_{\ell^{p,\eps}}}
\nc{\rar}{\rightarrow}
\nc{\sgn}{{\rm sign}}
\date{\today}

\begin{document}

\title{Nearly linear dynamics of nonlinear dispersive waves}

\author{M.~B.~Erdo\smash{\u{g}}an, N.~Tzirakis, and V.~Zharnitsky}
\thanks{The authors were partially supported by NSF grants DMS-0900865 (B.~E.), DMS-0901222 (N.~T.), and DMS-0807897  (V.~Z.)}

\address{Department of Mathematics \\
University of Illinois \\
Urbana, IL 61801, U.S.A.}

\email{berdogan@uiuc.edu \\ tzirakis@math.uiuc.edu\\ vzh@uiuc.edu}

\maketitle

\begin{abstract}
Dispersive averaging effects are used to show that KdV equation with periodic boundary conditions 
possesses 
high frequency solutions which behave nearly linearly. Numerical simulations are presented which 
indicate high 
accuracy of this approximation. Furthermore, this result is applied to shallow 
water wave dynamics in the limit of KdV approximation, which is obtained by asymptotic analysis 
in combination with  
numerical simulations of KdV.
\end{abstract}

\section{Introduction}
The study of the dynamics of high frequency waves, see e.g. \cite{burak-vadim}, has been motivated by 
the so-called quasilinear 
phenomenon in optical communication, where it was observed that spatially localized pulses evolve 
nearly linearly.
The dynamics of high frequency waves can also be motivated  by the well-posedness results in spaces 
of low regularity  for dispersive
PDEs, see {\em e.g.} \cite{bo1, bo2, ckstt}. These papers indicated that a subtle high frequency 
averaging effect took place in the
nonlinear dispersive dynamics making these results possible.

More recently, in \cite{our-arxiv} and \cite{bit}, KdV was studied with regard to this averaging effect. 
In \cite{our-arxiv},
near-linear dynamics was established for high frequency initial data and in \cite{bit} 
a new elegant proof of well-posedness
 in  $H^s, s\geq 0$ was found using explicitly high frequency averaging effects. 

The purpose of this article is twofold. First we establish  
 near-linear dynamics in KdV under weaker and more natural assumptions than \cite{our-arxiv}.
The proof relies 
on the so-called differentiation by parts technique (which is a variant of the normal 
form procedure) from \cite{bit}.
Secondly, we investigate how near-linear dynamics for KdV can be extended to  
the water waves problem. 
We use the standard derivation of KdV in the long wave, shallow water approximation 
to obtain physical parameters for which near-linear behavior might be observed.

%In this paper we present 
%a simpler proof of the main theorem in \cite{our-arxiv} with weaker and more natural assumptions. We use the 
%so-called differentiation by parts technique (which is a variant of normal form procedure) from \cite{bit}. 

We should note that for   KdV on the torus or a circle the linear solution is periodic in space and time and thus 
one does not have dispersive decay. It is also expected that the solutions of KdV on the torus will not be 
approximated by the linear evolution for infinite time. 
Therefore the proof that the nonlinear evolution is almost linear (in the sense of the subsequent Theorem) on 
a finite but large time scale in a way provides evidence that dispersion phenomena are not completely muted on 
the periodic setting. The seminal papers of Bourgain, \cite{bo1, bo2} by establishing Strichartz type estimates 
for periodic dispersive equations was probably the first step towards these new developments.

\begin{theo}\label{theo:main}
Consider the real valued zero mean solution of the KdV equation
$$u_t=u_{xxx}+uu_x$$
on $\mathbf{T}\times \R$ with the initial data $u(x,0)=\phi(x)$ satisfying
$$\|\phi\|_2=1, \,\,\,\,\,\,\,\|\phi\|_{H^{-1/2}}=\eps\ll 1.$$
Then, for each $t>0$ and small $\delta >0$, we have
\[
\|u(\cdot, t)-e^{t\partial_x^3}\phi\|_2 \leq  C_\delta\left ( \eps^2+t \, \eps^{1-\delta} \right ).
\]
\end{theo}
Note that the difference between the actual solution and the solution of the linear KdV is small in $L^2$
for $t$ up to $\epsilon^{-1+}$. 

We note that near-linear dynamics for high frequency solutions is easier to establish on 
unbounded domains, as the solution disperses to infinity and weakly nonlinear theories apply.
On bounded domains ({\em i.e.} with periodic boundary conditions) the solutions cannot scatter to infinity. 
For the NLS in the 2d torus,  Colliander {\em et al.}, \cite{ckstt1}, gave recently a nice proof. 
\mbox{Theorem \ref{theo:main}} is an intermediate 
result between the iterated linear solutions and the situation at infinity. The nonlinearity averages out 
since dispersion will cause high harmonics to oscillate rapidly. Throughout this paper we assume that we 
have global well-posed solutions which have additional regularity properties. The interested reader should see \cite{bo2} for the details. 
In addition, smooth solutions of KdV satisfy momentum conservation:
\\
$$\int_{-\pi}^{\pi}u(x,t)dx=\int_{-\pi}^{\pi}u(x,0)dx.$$ 
\\
Because of the momentum conservation law we can modify the equation adding a harmless term and thus only consider mean zero solution. This will imply the the Fourier series representation for the solution will have nonzero Fourier modes, an assumption that we consistently make in our paper. In particular notice that all the norms are restricted to this subclass of smooth solutions. In addition we use the conservation of energy,\\
$$\int_{-\pi}^{\pi}u^{2}(x,t)dx=\int_{-\pi}^{\pi}u^{2}(x,0)dx.$$
\\
%$H(u)=\int_{-\pi}^{\pi}\left( \frac{1}{2}u_{x}^{2}(x,t)+u^{3}(x,t)\right)dx=H(\phi)$ (Hamiltonian conservation)
The KdV equation is locally well-posed (LWP) in $L^{2}(\Bbb T)$, \cite{bo2}. Due to energy conservation KdV is globally well-posed and $u \in C(\Bbb R;L^{2}(\Bbb T)$). Kenig, Ponce, Vega, \cite{kpv}, improved Bourgain's result and showed that the solution of the KdV is LWP in $H^{s}(\Bbb T)$ for any $s>-\frac{1}{2}$. Later, Colliander, Keel, Staffilani, Takaoka, Tao, \cite{ckstt}, showed that the KdV is globally well-posed in
$H^{s}(\Bbb T)$ for any $s \geq -\frac{1}{2}$ thus adding a LWP result for the endpoint $s=-\frac{1}{2}$. Recently T. Kappeler and P. Topalov, \cite{kt} extended the latter result and prove that the KdV is globally well-posed in
$H^{s}(\Bbb T)$ for any $s \geq -1$.
\vskip 0.2in
The main idea of our proof runs as follows. First write KdV
$$u_t=u_{xxx}+uu_x $$
on the Fourier side,
\\
$$\partial_t u_k=\frac{ik}{2}   \sum_{k_1+k_2=k}u_{k_1}u_{k_2}- ik^3u_k,\,\,\,\,\,\,u_k(0)=\widehat\phi(k).$$
\\
Then, using the identity
\\
$$(k_1+k_2)^3-k_1^3-k_2^3=3(k_1+k_2)k_1k_2,$$ 
\\
and the transformation
\\
$$v_k(t)=u_k(t)e^{ik^3t}$$
\\
the equation can be written in the form
\\
$$\partial_t v_k=\frac{ik}{2}   \sum_{k_1+k_2=k}e^{i3kk_1k_2t}v_{k_1}v_{k_2}.$$
\\
The substitution
\\
$$u_{k}(t)=e^{-itk^3}v_k(t)$$
\\
eliminates the linear term $ik^3$ which has the highest growth at infinity and introduces oscillating exponentials into the nonlinear term. We have to show that $ v_k$ stays almost constant for large times under our high frequency assumption. We cannot neglect the averaging effects of the exponent. Without the exponential factor the above system corresponds to Burger's equation which exhibits strongly nonlinear dynamics. Informally speaking, we mainly have three types of terms:
\\
\\
a) Low frequency harmonics, with $n_1, n_2$ small, give negligible contributions because of the high frequency assumption.
\\
\\
b) Intermediate terms are few in number as the Diophantine equation $n_1n_2n=N$ has few solutions.
\\
\\
c) High frequency harmonics, with $n_1, n_2$ large, are well-averaged by the exponent.
\\
\\
We already mentioned that the method we use was inspired by \cite{bit}. Originally it was developed by Babin, Mahalov and Nicolaenko, \cite{bmn, bmn1}, in studying the global regularity of solutions of 3D problems in hydrodynamics (Navier-Stokes or Boussinesq system). In their framework the presence of high-frequency waves lead to destructive interference and weakened the nonlinearity through time averaging allowing one to prove global regularity. For the KdV the high Fourier modes of the linear term generates high-frequency oscillations which make the nonlinearity milder. There is an analogous phenomenon with the propagation of regularity to the Burger's equation with fast rotation
\\
$$u_{t}+uu_{x}=i\Omega u\ \ \ \ u(x,0)=\phi(x).$$
\\
Using Duhamel's formula the solution can be written as
\\
$$u(x,t)=\phi(x)-\int_{0}^{t}e^{i\Omega t}u(x,s)u_{x}(x,s)ds.$$
\\
For large $|\Omega|$ the nonlinearity weakens and the life-span of the solution is prolonged. So large oscillations is what separates the bad behavior of the classical Burger's equation and the good behavior of the Burger's equation with fast rotation. The same method has recently been applied by Kwon and Oh, \cite{ko}, to prove unconditional well-posedness for the modified KdV. The method of differentiation by parts helps to establish a priori estimates only in the $C_{t}^{0}H_{x}^{s}$ norms for any $s \geq \frac{1}{2}$. This is the heart of the matter in proving unconditional uniqueness, that is uniqueness of solutions to the modified KdV equation in the space $C_{t}^{0}H_{x}^{s}$ alone.
\vskip 0.2in
The second motivation for our work comes from the various connections of the high frequency averaging process that
 we describe with certain aspects of the water wave theory. The dynamics of surface water waves has been 
an important object of study in science for over a century. Soliton solutions
and integrability in P.D.E.'s are two examples of remarkable discoveries that were made by
investigating  water wave dynamics in shallow waters. In more recent times,  the so-called rogue waves
have been under an intense investigation, see  for example \cite{KharifPelin2,osborne,zakharov} and the
references therein.
These unusually large waves have been observed in various parts of the ocean in both deep, 
see {\em e.g.} \cite{KharifPelin}, and shallow water, see {\em e.g.} \cite{sand}, motivating scientists 
to suggest various mechanisms for rogue wave formation.

In the case of shallow water, one normally does not work with the full water wave equation but uses
approximate models to study the evolution, in particular the formation of rogue waves.
These models are nonlinear dispersive
equations such as KdV, Boussinesq approximations, {\em etc}. In particular, KdV describes
unidirectional small
amplitude long waves on fluid surface. See, {\em e.g.} \cite{talipova} for applications of KdV to
rouge waves in shallow water. Since rouge waves correspond to concentration of energy on small
domains, one might argue that higher frequencies play important role in rouge waves formation.

Here we provide some evidence, based on asymptotic expansions and numerical simulations that for sufficiently  high frequency initial data, one-dimensional
spatially periodic surface waves in  shallow water exhibit near-linear behavior. Thus, linear theories of rogue wave formations can be extended
to nonlinear high frequency regime.

Clearly, one has to be careful when considering short wave solutions for the equations obtained in
the long wave approximations such as KdV. However, we show that there is a  set of
parameters when our high frequency solutions correspond to a realistic physical scenario
in shallow  water waves, see \mbox{Section \ref{sec:numerics}}.

\section{Normal form reduction using ``differentiation by parts''}
In this section we apply a variant of normal form reduction, called differentiation by parts \cite{bit}, to bring the equation to a more convenient form in which low order resonant terms are separate from the other terms.
 Using the Fourier series representation $$u(x,t)=\sum_{k \in \Bbb Z_{0}}u_{k}(t)e^{ikx}$$ with
\\ $$u_k:=\widehat{u}(k)=\frac{1}{2\pi}\int_0^\pi u(t,x)e^{-ikx} dx$$
\\
we express KdV as an infinite system of ordinary differential equations
\\
$$\partial_t u_k=\frac{ik}{2}   \sum_{k_1+k_2=k}u_{k_1}u_{k_2}- ik^3u_k,\,\,\,\,\,\,u_k(0)=\widehat\phi(k).$$
\\
Notice that since the solution is real valued we have that $\bar{u}_{k}=u_{-k}$. Changing the variable 
\\
$$v_k(t)=u_k(t)e^{ik^3t}$$
\\ (notice again that $\bar{v}_{k}=v_{-k}$), and using the identity
\\
$$(k_1+k_2)^3-k_1^3-k_2^3=3(k_1+k_2)k_1k_2,$$ 
\\we obtain
\begin{align}\label{vkdv}
 \partial_t v_k=\frac{ik}{2}   \sum_{k_1+k_2=k}e^{i3kk_1k_2t}v_{k_1}v_{k_2},\,\,\,\,\,\,v_k(0)=\widehat\phi(k).
\end{align}
\\
Since $e^{i3kk_1k_2t}=\partial_{t}( \frac{1}{3ikk_1k_2}e^{i3kk_1k_2t})$ differentiation by parts and \eqref{vkdv} yields
\\
$$\partial_{t}v_{k}=\partial_{t}\left( \frac12 ik \sum_{k_1+k_2=k}\frac{e^{3ikk_1k_2t}v_{k_1}v_{k_2}}{3ikk_1k_2}\right)-\frac12 ik\sum_{k_1+k_2=k}\frac{e^{3ikk_1k_2t}}{3ikk_1k_2}\partial_{t}(v_{k_1}v_{k_2})=$$
\\
$$\frac16\partial_{t}\left(\sum_{k_1+k_2=k}\frac{e^{3ikk_1k_2t}v_{k_1}v_{k_2}}{k_1k_2}\right)-\frac16 \sum_{k_1+k_2=k}\frac{e^{3ikk_1k_2t}}{k_1k_2}(\partial_{t}v_{k_1}v_{k_2}+\partial_{t}v_{k_2}v_{k_1}).$$
\\
Note that since $v_0=0$, in the sums above $k_1$ and $k_2$ are not zero.
The last two terms are symmetric with respect to $k_1$ and $k_2$ and thus we can consider only one of them. Using \eqref{vkdv} we have
\\
$$\sum_{k_1+k_2=k}\frac{e^{3ikk_1k_2t}}{k_1k_2}v_{k_1}\partial_{t} v_{k_2}=\frac12 i\sum_{k=k_1+k_2}\frac{e^{3ikk_1k_2t}}{k_1}v_{k_1}\left( \sum _{\mu+\lambda=k_2}e^{3itk_2\mu\lambda}v_{\mu}v_{\lambda}\right)=$$
\\
$$\frac12 i\sum_{k=k_1+\mu+\lambda}\frac{v_{k_1}v_{\mu}v_{\lambda}}{k_1}e^{3it[kk_1(\mu+\lambda)+\mu\lambda(\mu+\lambda)]}.$$
\\
We note that $\mu+\lambda$ can not be zero since $\mu+\lambda=k_2$. Using the identity 
\\
$$kk_1+\mu\lambda=(k_1+\mu+\lambda)k_1+\mu\lambda=(k_1+\mu)(k_1+\lambda)$$
\\ and thus by renaming the variables $k_2=\mu, k_3=\lambda$, we have that
\\
$$\sum_{k_1+k_2=k}\frac{e^{3ikk_1k_2t}}{k_1k_2}v_{k_1}\partial_{t} v_{k_2}=\frac12i\sum_{\stackrel{k_1+k_2+k_3=k}{k_2+k_3\neq 0}}\frac{e^{3it(k_1+k_2)(k_2+k_3)(k_3+k_1)}}{k_1}v_{k_1}v_{k_2}v_{k_3}.$$
\\
All in all we have that
\\
$$\partial_t\left(v_k-\frac16B_2(v,v)_k\right)=\frac{i}{6}R_3(v,v,v)_k$$
\\
where
\\
$$B_2(u,v)_k=\sum_{k_1+k_2=k}\frac{e^{3ikk_1k_2t}u_{k_1}v_{k_2}}{k_1k_2}$$
\\
and
\\
$$R_3(u,v,w)_k=\sum_{\stackrel{k_1+k_2+k_3=k}{k_2+k_3\neq 0}}\frac{e^{3it(k_1+k_2)(k_2+k_3)(k_3+k_1)}}{k_1}u_{k_1}v_{k_2}w_{k_3}.$$
\\
Now let's single out the terms (resonant terms) for which
\begin{equation}\label{zeroset}
(k_1+k_2)(k_3+k_1)=0
\end{equation}
 and write
\\
$$R_3(v,v,v)_k=R_{3r}(v,v,v)_k+R_{3nr}(v,v,v)_{k}$$
\\
where the subscript $r$ and $nr$ stands for the resonant and non-resonant terms respectively. Thus
\\
$$R_{3r}(v,v,v)_k=\sum_{\stackrel{k_1+k_2+k_3=k}{k_2+k_3\neq 0}}^{r}\frac{v_{k_1}v_{k_2}v_{k_3}}{k_1}$$
\\
and
\\
$$R_{3nr}(v,v,v)_k=\sum_{k_1+k_2+k_3=k}^{nr}\frac{e^{3it(k_1+k_2)(k_2+k_3)(k_3+k_1)}}{k_1}v_{k_1}v_{k_2}v_{k_3}.$$
\\
The set for which \eqref{zeroset} holds is the disjoint union of the following 3 sets
\\
$$S_{1}=\{k_1+k_2=0\}\cap\{k_3+k_1=0\}\Leftrightarrow \{k_1=-k,\ k_2=k,\ k_3=k\},$$
\\
$$S_{2}=\{k_1+k_2=0\} \cap \{ k_3+k_1\ne 0\} \Leftrightarrow \{k_1=j,\ k_2=-j,\ k_3=k,\ |j| \neq k\},$$
\\
$$S_{3}=\{k_3+k_1=0\}\cap\{k_1+k_2\ne 0\} \} \Leftrightarrow \{k_1=j,\ k_2=k,\ k_3=-j,\ |j| \neq k\}.$$
\\
Thus
\\
$$R_{3r}(v,v,v)_k=\sum_{\lambda=1}^{3}\sum_{S_{\lambda}}\frac{v_{k_1}v_{k_2}v_{k_3}}{k_1}=
\frac{v_{-k}v_{k}v_{k}}{-k} + v_k\sum_{\stackrel{j\in \Bbb Z_{0}}{ |j|\neq k}}\frac{v_{j}v_{-j}}{j} +v_k\sum_{\stackrel{j\in \Bbb Z_{0}}{ |j|\neq k}}\frac{v_{j}v_{-j}}{j}.$$
\\
Note that the second and third terms in the sum above  are identically zero due to the symmetry relation $j \leftrightarrow -j$. Thus
\\
$$R_{3r}(v,v,v)_k=-\frac{v_k}{k}|v_k|^2.$$
\\
We obtain
\\
$$\partial_{t}\left( v_k-\frac16 B_2(v,v)_k\right)=-\frac{i}{6k}v_k|v_k|^2+\frac{i}{6}R_{3nr}(v,v,v)_k.$$
\\
Since the exponent in the last term is not zero we can differentiate by parts one more time and obtain that
\\
$$R_{3nr}(v,v,v)_k=\sum_{k_1+k_2+k_3=k}^{nr}\frac{e^{3it(k_1+k_2)(k_2+k_3)(k_3+k_1)}}{k_1}v_{k_1}v_{k_2}v_{k_3}=$$
\\
$$\frac{1}{3i}\partial_{t}B_{3}(v,v,v)_k-\frac{1}{3i}
\sum_{k_1+k_2+k_3=k}^{nr}\frac{e^{3it(k_1+k_2)(k_2+k_3)(k_3+k_1)}}{k_1(k_1+k_2)(k_2+k_3)(k_3+k_1)}\times$$
\\
$$\left( \partial_{t}v_{k_1}v_{k_2}v_{k_3}+\partial_{t}v_{k_2}v_{k_1}v_{k_3}+\partial_{t}v_{k_3}v_{k_1}v_{k_2}\right)$$
\\
where
\\
$$B_{3}(u,v,w)_k=\sum_{k_1+k_2+k_3=k}^{nr}\frac{e^{3it(k_1+k_2)(k_2+k_3)(k_3+k_1)}}{k_1(k_1+k_2)(k_2+k_3)(k_3+k_1)}u_{k_1}v_{k_2}w_{k_3}.$$
\\
As before we express time derivatives using \eqref{vkdv}. The terms containing $\partial_{t}v_{k_2}$ and $\partial_{t}v_{k_3}$ produce the same expressions and a
calculation reveals that
\\
$$\sum_{k_1+k_2+k_3=k}^{nr}\frac{e^{3it(k_1+k_2)(k_2+k_3)(k_3+k_1)}}{k_1(k_1+k_2)(k_2+k_3)(k_3+k_1)}\times$$
\\
$$\left(\partial_{t}v_{k_1}v_{k_2}v_{k_3}+\partial_{t}v_{k_2}v_{k_1}v_{k_3}+\partial_{t}v_{k_3}v_{k_1}v_{k_2}\right)=iB_{4}(v,v,v,v)_{k}$$
\\
where
\\
$$B_{4}(u,v,w,z)_{k}=\frac12B_4^1(u,v,w,z)_{k}+B_{4}^{2}(u,v,w,z)_{k}.$$
\\
From now on $\sum^*$ means that the sum is over all indices for which the denominator do not vanish.The term corresponding to $\partial_{t}v_{k_1}$ is
\\
$$B_4^1(u,v,w,z)_{k}=\sum_{k_1+k_2+k_3+k_4=k}^{\star}\frac{e^{3it\psi(k_1,k_2,k_3,k_4)}}{(k_1+k_2)(k_1+k_3+k_4)(k_2+k_3+k_4)}u_{k_1}v_{k_2}w_{k_3}z_{k_4},$$
\\
and the sum of the terms corresponding to $\partial_{t}v_{k_2}$ and $\partial_{t}v_{k_3}$ is
\\
$$B_4^2(u,v,w,z)_{k}=\sum_{k_1+k_2+k_3+k_4=k}^{\star}\frac{e^{3it\psi(k_1,k_2,k_3,k_4)}(k_3+k_4)}{k_1(k_1+k_2)(k_1+k_3+k_4)(k_2+k_3+k_4)}u_{k_1}v_{k_2}w_{k_3}z_{k_4}.$$
\\
The phase function $\psi$ will be irrelevant for our calculations since it is going to be estimated out by taking absolute values inside the sums. For completeness we note that it can be expressed as
\\
$$(k_1+k_2+k_3+k_4)^2-k_1^3-k_2^3-k_3^3-k_4^3.$$
\\
Hence for $R_{3nr}(v,v,v)_k$ we have:
\\
$$R_{3nr}(v,v,v)_k=\frac{1}{3i}\partial_{t}B_3(v,v,v)_k-\frac13\left(\frac12B_4^1(v,v,v,v)_{k}+B_{4}^{2}(v,v,v,v)_{k}\right).$$
\\
If we put everything together and combining the two $B_4$ terms in one we obtain
\\
\begin{equation}\label{b4eq}
\partial_t\big(v_k-\frac16 B_2(v,v)_k - \frac{1}{18}B_3(v,v,v)_k\big)
= -\frac{iv_k|v_k|^2}{6k}+\frac{i}{18}B_4(v,v,v,v)_k,
\end{equation}
where
\\
$$
B_2(v)_k =\sum_{k_1+k_2=k}\frac{e^{i3kk_1k_2t}v_{k_1}v_{k_2}}{k_1k_2}
$$
\\
$$
B_3(v)_k =\sum^*_{k_1+k_2+k_3=k}\frac{e^{i3(k_1+k_2)(k_1+k_3)(k_2+k_3)t}v_{k_1}v_{k_2}v_{k_3}}
{k_1(k_1+k_2)(k_1+k_3)(k_2+k_3)}
$$
\\
$$
B_4(v)_k =\frac12\sum^*_{k_1+k_2+k_3+k_4=k}\frac{e^{i\psi(k_1,k_2,k_3,k_4)t}(2k_3+2k_4+k_1)v_{k_1}v_{k_2}v_{k_3}v_{k_4}}
{k_1(k_1+k_2)(k_1+k_3+k_4)(k_2+k_3+k_4)},
$$
\\
where $\psi(k_1,k_2,k_3,k_4)=(k_1+k_2+k_3+k_4)^3-k_1^3-k_2^3-k_4^3$.

\section{Proofs}
\noindent
{\em Notation:} 
To avoid the use of multiple constants, we  write $A \lesssim B$ to denote that there is an absolute  constant 
$C$ such that $A\leq CB$.  We will also use frequently the notation 
$A \lesssim B(\eta-)$ if for any $\gamma>0$, $A \leq C_{\gamma}B(\eta-\gamma)$. Similar notation 
will be used  for $A \lesssim B(\eta+)$.
Finally, for $s\in\R$, we define the homogeneous Sobolev norm
\[
\|u\|_{\dot{H}^s}=\Big(\sum_{k\neq 0} |k|^{2s}|u_k|^2\Big)^{1/2}.
\]
We have
\begin{prop}\label{prop:apriori}
The following a-priori estimates hold
\be\label{p1}
\|B_2(v)\|_{\dot{H}^{-1/2}}\leq \|B_2(v)\|_{L^2} \lesssim \|v\|_{\dot{H}^{-1/2}}^2,
\ee
\be\label{p2}
\|B_3(v)\|_{\dot{H}^{-1/2}}\leq \|B_3(v)\|_{L^2} \lesssim \|v\|_{\dot{H}^{-1/2}}^2 \|v\|_{L^2},
\ee
\be\label{p3}
\|B_4(v)\|_{L^2} \lesssim \|v\|_{\dot{H}^{-1/2}}^{1-} \|v\|_{L^2}^{3+},
\ee
\be\label{p4}
\|B_4(v)\|_{\dot{H}^{-1/2}}\lesssim \|v\|_{\dot{H}^{-1/2}}^{2-} \|v\|_{L^2}^{2+},
\ee
\be \label{p5}
\|v_k^3/k\|_{\ell^2}\lesssim \|v\|_{\dot{H}^{-1/2}}^{2} \|v\|_{L^2}.
\ee
\end{prop}

Now we will prove Theorem~\ref{theo:main} using Proposition~\ref{prop:apriori}.
\begin{proof}[Proof of Theorem~\ref{theo:main}]
First note that
\\
$$
\|u(\cdot, t)-e^{t\partial_x^3}\phi\|_2=\|u_k(t)-e^{ik^3t}u_k(0)\|_{\ell^2(k)}=\|v_k(t)-v_k(0)\|_{\ell^2(k)}.
$$
\\
We will estimate the R.H.S. using \eqref{b4eq}.
Integrating \eqref{b4eq} from $0$ to $T$ we have
\begin{multline}\nonumber
v_k(T)-v_k(0)= \frac16 B_2(v(T))-\frac16 B_2(v(0))+\frac{1}{18} B_3(v(T)) - \frac{1}{18} B_3(v(0))   \\
 -\frac{i}{6} \int_0^T \frac{v_k|v_k|^2}{k}dt+\frac{i}{18}\int_0^T B_4(v)_k dt.
\end{multline}
The estimates in Proposition~\ref{prop:apriori}, and the fact that for each $t$, $\|v(t)\|_{L^2}\lesssim 1$, imply that
\begin{align}\label{L2est}
\|v(T)-v(0) \|_{L^2}&\lesssim   \|v(T)\|_{\dot{H}^{-1/2}}^2+ \|v(0)\|_{\dot{H}^{-1/2}}^2 \\ \\&
+\int_0^T \big( \|v(t)\|_{\dot{H}^{-1/2}}^2+ \|v(t)\|_{\dot{H}^{-1/2}}^{1-}\big) dt,\nonumber
\end{align}
\vskip 0.1in
\begin{align}\label{Hsest}
\|v(T)-v(0) \|_{\dot{H}^{-1/2}}&\lesssim   \|v(T)\|_{\dot{H}^{-1/2}}^2+ \|v(0)\|_{\dot{H}^{-1/2}}^2  \\ \\&+\int_0^T \big( \|v(t)\|_{\dot{H}^{-1/2}}^2+ \|v(t)\|_{\dot{H}^{-1/2}}^{2-}\big) dt.\nonumber
\end{align}
Since $\|v(0)\|_{\dot{H}^{-1/2}}=\|\phi\|_{\dot{H}^{-1/2}}<\eps$, the inequality \eqref{Hsest} and the continuity of the solution in $L^2$ (and hence in $\dot{H}^{-1/2}$) imply that, for all $T\lesssim \eps^{-1+}$, $\|v(T)\|_{\dot{H}^{-1/2}} \lesssim \eps.$
Using this in \eqref{L2est} implies that for $T\lesssim \eps^{-1+}$
\\
$$
\|v(T)-v(0) \|_{L^2}\lesssim \eps^2+T\eps^{1-}.
$$
\end{proof}

\begin{proof}[Proof of Proposition~\ref{prop:apriori}]
We start with \eqref{p5}. Using $v_0=0$
\\
$$
\|v_k^3/k\|_{\ell^2}\leq \|v_k/\sqrt{k}\|_{\ell^\infty}^2 \|v_k\|_{\ell^2}\leq \|v_k/\sqrt{k}\|_{\ell^2}^2 \|v_k\|_{\ell^2}= \|v\|_{\dot{H}^{-1/2}}^{2} \|v\|_{L^2}.
$$
\\
We continue with \eqref{p1}. It suffices to estimate $B_2$ in $L^2$:
\begin{align*}
\|B_2(v)\|_{L^2} &=\Big\|\sum_{k_1+k_2=k}\frac{e^{i3kk_1k_2t}v_{k_1}v_{k_2}}{k_1k_2}\Big\|_{\ell^2}
\leq \big\|\frac{|v_k|}{k}*\frac{|v_k|}{k}\big\|_{\ell^2}\\
\\
&\lesssim \|v_k/k\|_{\ell^{4/3}}^2 \lesssim \|v\|_{\dot{H}^{-1/2}}^2,
\end{align*}
where we used the Young and H\"older inequalities in the second and third inequalities respectively. Now consider \eqref{p3}:\\
$$
\|B_3(v)\|_{L^2}^2 =\Big\|\sum^*_{k_1+k_2+k_3=k}\frac{e^{i3(k_1+k_2)(k_1+k_3)(k_2+k_3)t}v_{k_1}v_{k_2}v_{k_3}}
{k_1(k_1+k_2)(k_1+k_3)(k_2+k_3)}\Big\|_{\ell^2}^2
$$
\\
$$\leq \Big\|\sum^*_{k_1+k_2+k_3=k}\frac{\sqrt{|k_2|}}
{\sqrt{|k_1|}|k_1+k_2||k_1+k_3||k_2+k_3|}\frac{|v_{k_1}|}{\sqrt{|k_1|}}\frac{|v_{k_2}|}{\sqrt{|k_2|}}|v_{k_3}|\Big\|_{\ell^2}^2
$$
\\
By Cauchy Schwarz we estimate this by
\\
$$\sum_k \Big(\sum^*_{k_1+k_2+k_3=k}\frac{|k_2|}
{|k_1||k_1+k_2|^2|k_1+k_3|^2|k_2+k_3|^2}\Big)\Big(\sum^*_{n_1+n_2+n_3=k}\frac{|v_{n_1}|^2}{|n_1|}\frac{|v_{n_2}|^2}{|n_2|}|v_{n_3}|^2\Big)
$$
\\
$$\leq \sup_k \Big(\sum^*_{k_1+k_2+k_3=k}\frac{|k_2|}
{|k_1||k_1+k_2|^2|k_1+k_3|^2|k_2+k_3|^2}\Big)\Big(\sum^*_{n_1,n_2,n_3}\frac{|v_{n_1}|^2}{|n_1|}\frac{|v_{n_2}|^2}{|n_2|}|v_{n_3}|^2\Big)
$$
\\
$$= \|v\|_{\dot{H}^{-1/2}}^4 \|v\|_{L^2}^2 \,\,\, \sup_k  \sum^*_{k_1+k_2+k_3=k}\frac{|k_2|}
{|k_1||k_1+k_2|^2|k_1+k_3|^2|k_2+k_3|^2}.
$$
\\
It remains to show that the supremum above is finite. Note that the supremum is $\lesssim$
\begin{align*}
\sup_k \sum^*_{k_1,k_2}\frac{|k_2|}
{|k_1||k_1+k_2|^2|k-k_1|^2|k-k_2|^2}&\lesssim \sum^*_{k_1,k_2}\frac{|k_2|}
{|k_1||k_1+k_2|^2 |k_1-k_2|^2}
\end{align*}
where we used the fact that, for $k_1,k_2\neq k$, $|k-k_1||k-k_2|\gtrsim |k_1-k_2|$.
Now to estimate this sum, consider the cases $|k_1|>2|k_2|$, $|k_1|<|k_2|/2$, and $|k_1|\approx|k_2|$ separately. In the first case, the sum is $\lesssim$
\\
$$\sum_{|k_1|>2|k_2| }^* \frac{1}{|k_2|^2 |k_1|^2 }<\infty.
$$
\\
In the second case, we have
\\
$$\sum_{|k_1|<|k_2|/2 }^* \frac{1}{|k_2|^3 |k_1| }\lesssim \sum_{ k_2  }^* \frac{\log(|k_2|)}{|k_2|^3 }<\infty.
$$
\\
In the third case we have
\\
$$\sum_{|k_1|\approx|k_2| }^* \frac{1}{|k_1+k_2|^2 |k_1-k_2|^2 }\lesssim \sum_{n_1,n_2  }^* \frac{1}{n_1^2 n_2^2 }<\infty.
$$
\\
Finally, we consider $B_4$. First note that
\begin{align*}
|B_4(v)_k|&\lesssim \sum^*_{k_1+k_2+k_3+k_4=k}\frac{|v_{k_1}v_{k_2}v_{k_3}v_{k_4}|}
{|k_1+k_2||k_1+k_3+k_4||k_2+k_3+k_4|}\\&+ \sum^*_{k_1+k_2+k_3+k_4=k}\frac{|v_{k_1}v_{k_2}v_{k_3}v_{k_4}|}
{|k_1||k_1+k_2| |k_2+k_3+k_4|}=:B_4^1(v)_k+B_4^2(v)_k.
\end{align*}
First we consider the $L^2$ norm of $B_4^1$. Applying Cauchy Schwarz as in the case of $B_3$, we have
\begin{align*}
&\|B_4^1(v)\|_{L^2}^2 \lesssim \\&\Big\|\sum^*_{k_1+k_2+k_3+k_4=k}\frac{|k_1|^{\frac12-}|v_{k_4}|}
{|k_1+k_2||k_1+k_3+k_4|^{\frac12-}|k_2+k_3+k_4|}\frac{|v_{k_1}v_{k_2}v_{k_3}| }
{|k_1|^{\frac12-}|k_1+k_3+k_4|^{\frac12+}}\Big\|_{\ell^2}^2\\
&
\leq\sup_k\sum^*_{k_1+k_2+k_3+k_4=k}\frac{|k_1|^{1-}|v_{k_4}|^2}
{|k_1+k_2|^2|k_1+k_3+k_4|^{1-}|k_2+k_3+k_4|^2}\\
&\,\,\,\,\,\,\,\,\,\,\,\Big(\sum^*_{n_1,n_2,n_3,n_4}\frac{|v_{n_1}v_{n_2}v_{n_3}|^2}
{|n_1|^{1-}|n_1+n_3+n_4|^{1+}}\Big).
\end{align*}
Note that the sum in the parenthesis is $\lesssim \|v\|^2_{\dot{H}^{-1/2+}}\|v\|_{L^2}^4$ (by summing in $n_4$ first). We estimate the supremum by eliminating $k_3$ in the sum as follows
\begin{multline}\label{123}
\sup_k\sum^*_{k_1,k_2,k_4}\frac{|k_1|^{1-}|v_{k_4}|^2}
{|k_1+k_2|^2|k-k_2|^{1-}|k-k_1|^2}\\=\|v\|_{L^2}^2\sup_k\sum^*_{k_1,k_2 }\frac{|k_1|^{1-} }
{|k_1+k_2|^2|k-k_2|^{1-}|k-k_1|^2}
\end{multline}
Using $|k_1+k_2|^2|k-k_2|^{1-}|k-k_1|^2\gtrsim |k_1|^{1-}|k_1+k_2|^{1+}|k-k_1|^{1+}$ we have
\\
$$\eqref{123}\lesssim \|v\|_{L^2}^2\sup_k\sum^*_{k_1,k_2  }\frac{1}
{|k_1+k_2|^{1+}|k-k_1|^{1+}}\lesssim \|v\|_{L^2}^2.
$$
\\
The last inequality follows by summing first in $k_2$ then in $k_1$. Now consider the $L^2$ norm of $B_4^2$. Similarly, we obtain
\begin{align*}
\|B_4^2(v)\|_{L^2}^2 &\lesssim  \sup_k\sum^*_{k_1+k_2+k_3+k_4=k}\frac{|v_{k_1}|^2|v_{k_2}|^2|v_{k_3}|^2 }
{|k_1|^{1-} }\\&\,\,\,\,\,\,\,\, \Big(\sum^*_{n_1,n_2,n_3,n_4}\frac{|v_{n_4}|^2}
{|n_1|^{1+}|n_1+n_2|^{2}|n_2+n_3+n_4|^2}\Big).
\end{align*}
The sum in parenthesis is $\lesssim \|v\|_{L^2}^2$ by summing first in $n_3$, then $n_4$, then $n_2$, and then in $n_1$. Finally
\\
$$
\sup_k\sum^*_{k_1+k_2+k_3+k_4=k}\frac{|v_{k_1}|^2|v_{k_2}|^2|v_{k_3}|^2 }
{|k_1|^{1-} }\leq \sum^*_{k_1,k_2,k_3}\frac{|v_{k_1}|^2|v_{k_2}|^2|v_{k_3}|^2 }
{|k_1|^{1-}}=\|v\|_{L^2}^4\|v\|^2_{\dot{H}^{-1/2+}}.
$$
\\
Combining the estimates for $B_4^1$ and $B_4^2$, we obtain \eqref{p3}:
\\
$$
\|B_4(v)\|_{L^2}\lesssim \|v\|_{L^2}^3\|v\|_{\dot{H}^{-1/2+}}\lesssim \|v\|_{L^2}^{3+}\|v\|^{1-}_{\dot{H}^{-1/2}}.
$$
\\
It remains to prove \eqref{p4}. We start by estimating $B_4^1$. Applying Cauchy Schwarz as above we have
\begin{align*}
\|B_4^1(v)\|_{\dot{H}^{-1/2}}^2& \leq \sup_k\sum^*_{k_1+k_2+k_3+k_4=k}\frac{|k_1|^{1-}|k_2|^{1-}|v_{k_4}|^2}
{|k||k_1+k_2|^2|k_1+k_3+k_4|^{1-}|k_2+k_3+k_4|^2}\\
& \Big(\sum^*_{n_1,n_2,n_3,n_4}\frac{|v_{n_1}v_{n_2}v_{n_3}|^2}
{|n_1|^{1-}|n_2|^{1-}|n_1+n_3+n_4|^{1+}}\Big).
\end{align*}
Note that the sum in parenthesis $\lesssim  \|v\|_{L^2}^2\|v\|_{\dot{H}^{-1/2+}}^4$. Eliminating $k_3$ in the first sum we have
\begin{align*}
\|B_4^1(v)\|_{\dot{H}^{-1/2}}^2& \lesssim  \|v\|_{L^2}^2\|v\|_{\dot{H}^{-1/2+}}^4\sup_k\sum^*_{k_1,k_2,k_4}\frac{|k_1|^{1-}|k_2|^{1-}|v_{k_4}|^2}
{|k||k_1+k_2|^2|k-k_2 |^{1-}|k-k_1|^2}\\
&\lesssim \|v\|_{L^2}^4\|v\|_{\dot{H}^{-1/2+}}^4\sup_k\sum^*_{k_1,k_2 }\frac{|k_1|^{1-}|k_2|^{1-} }
{|k||k_1+k_2|^2|k-k_2 |^{1-}|k-k_1|^2}
\end{align*}
Note that the sum above is bounded by
\\
$$
\sum^*_{k_1,k_2 }\frac{ |k_2|^{1-} }
{ |k_1+k_2|^2|k-k_2 |^{1-}|k-k_1|^2}+\sum^*_{k_1,k_2 }\frac{ |k_2|^{1-} }
{|k||k_1+k_2|^2|k-k_2 |^{1-}|k-k_1|^{1+}}
$$
\\
$$
\lesssim \sum^*_{ k_2 }\frac{ |k_2|^{1-} }
{ |k+k_2|^2|k-k_2 |^{1-} }+\sum^*_{k_2 }\frac{ |k_2|^{1-} }
{|k||k +k_2|^{1+}|k-k_2 |^{1-}}\lesssim \sum^*_{k_2 }\frac{ 1}
{ |k +k_2|^{1+} }<\infty
$$
\\
In the first inequality we used  (for $a>b\geq 1$)
\be\label{abineq}
\sum_n^*\frac{1}{|n+m|^a|n|^b}\lesssim \frac{1}{|m|^{b}},
\ee
which follows by considering the cases $|n|<|m|/2$ and $|n|\geq |m|/2$ separately. In the second inequality
we used $|k+k_2||k-k_2|\gtrsim |k_2|$ and $|k||k-k_2|\gtrsim |k_2|$. Similarly,
\begin{align*}
\|B_4^2(v)\|_{\dot{H}^{-1/2}}^2& \leq \sup_k\sum^*_{k_1+k_2+k_3+k_4=k}\frac{|v_{k_1}|^2|v_{k_2}|^2 |v_{k_3}|^2}
{|k_1|^{1-}|k_2|^{1-}}\\
& \Big(\sum_k^*\sum^*_{n_1+n_2+n_3+n_4=k}\frac{|n_2|^{1-}|v_{n_4}|^2}
{|k||n_1|^{1+}|n_1+n_2|^2 |n_2+n_3+n_4|^2}\Big).
\end{align*}
Note that the supremum is $\lesssim  \|v\|_{L^2}^2\|v\|_{\dot{H}^{-1/2+}}^4$. Eliminating $n_3$ in the parenthesis we obtain
\begin{align*}
\|B_4^2(v)\|_{\dot{H}^{-1/2}}^2& \lesssim  \|v\|_{L^2}^2\|v\|_{\dot{H}^{-1/2+}}^4 \sum^*_{n_1,n_2,n_4,k}\frac{|n_2|^{1-}|v_{n_4}|^2}
{|k||n_1|^{1+}|n_1+n_2|^2 |k-n_1|^2}\\
&\lesssim \|v\|_{L^2}^4\|v\|_{\dot{H}^{-1/2+}}^4 \sum^*_{n_1,n_2,k }\frac{|n_2|^{1-} }
{|k||n_1|^{1+}|n_1+n_2|^2 |k-n_1|^2}
\end{align*}
Applying \eqref{abineq} to the sum in $k$, we have
\begin{align*}
\|B_4^2(v)\|_{\dot{H}^{-1/2}}^2& \lesssim \|v\|_{L^2}^4\|v\|_{\dot{H}^{-1/2+}}^4 \sum^*_{n_1,n_2  }\frac{|n_2|^{1-} }
{ |n_1|^{2+}|n_1+n_2|^2  }\lesssim \|v\|_{L^2}^4\|v\|_{\dot{H}^{-1/2+}}^4.
\end{align*}
The last inequality follows by applying \eqref{abineq} to the sum in $n_1$ and then summing in $n_2$.
This yields \eqref{p4}.
\end{proof}

\section{Numerical simulations demonstrating high accuracy of approximation}
In this section, we present numerical evidence that for the initial data with sufficiently high 
frequency, linear KdV approximates very well nonlinear KdV. As the initial condition we use
 the first Hermit function, appropriately scaled
\[
u(x) = \frac{1}{\sqrt{\epsilon}}\left ( \frac{x}{\epsilon} \right ) e^{-\frac{x^2}{2\epsilon^2}}.
\]

The Figures \ref{fig:spec2pi} and \ref{fig:phys2pi} show the initial and evolved waves in KdV with 
periodic boundary conditions $u(x+2\pi) = u(x)$ for the time interval $T = 2\pi$. Here $\epsilon =0.1$ and 
the time step is $\Delta t = 10^{-7}$.
Note that the linear KdV evolution is $2\pi$-periodic in time.  Therefore, if 
the near-linear dynamics takes place, we should see nearly perfect return of the evolved data to 
the original profile. Both figures confirm such behavior.

\begin{figure}[htd]
\includegraphics[width=80mm]{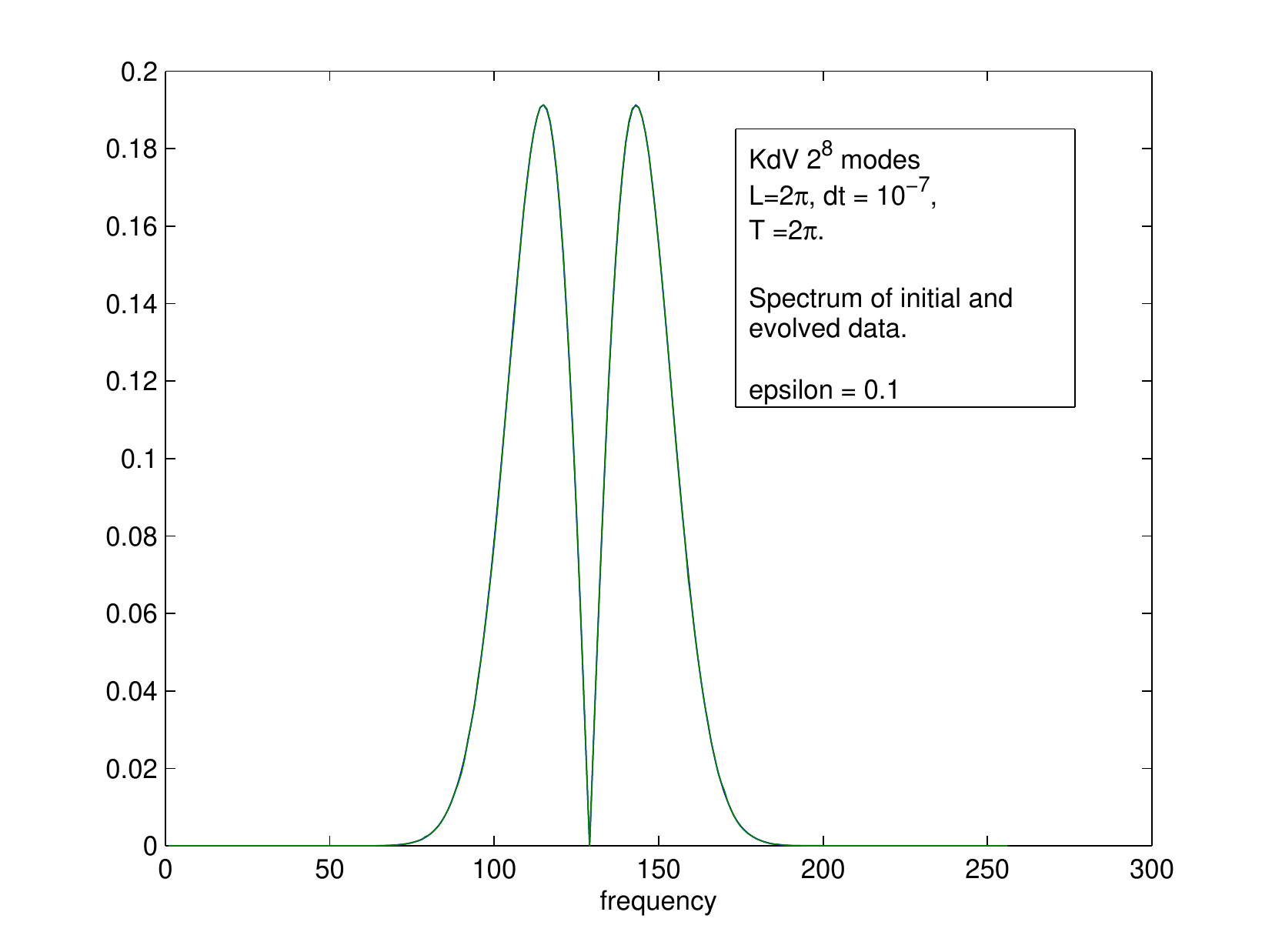}
\caption{Spectral data for initial and evolved waves.}
\label{fig:spec2pi}
\end{figure}

\begin{figure}[htd]
\includegraphics[width=100mm]{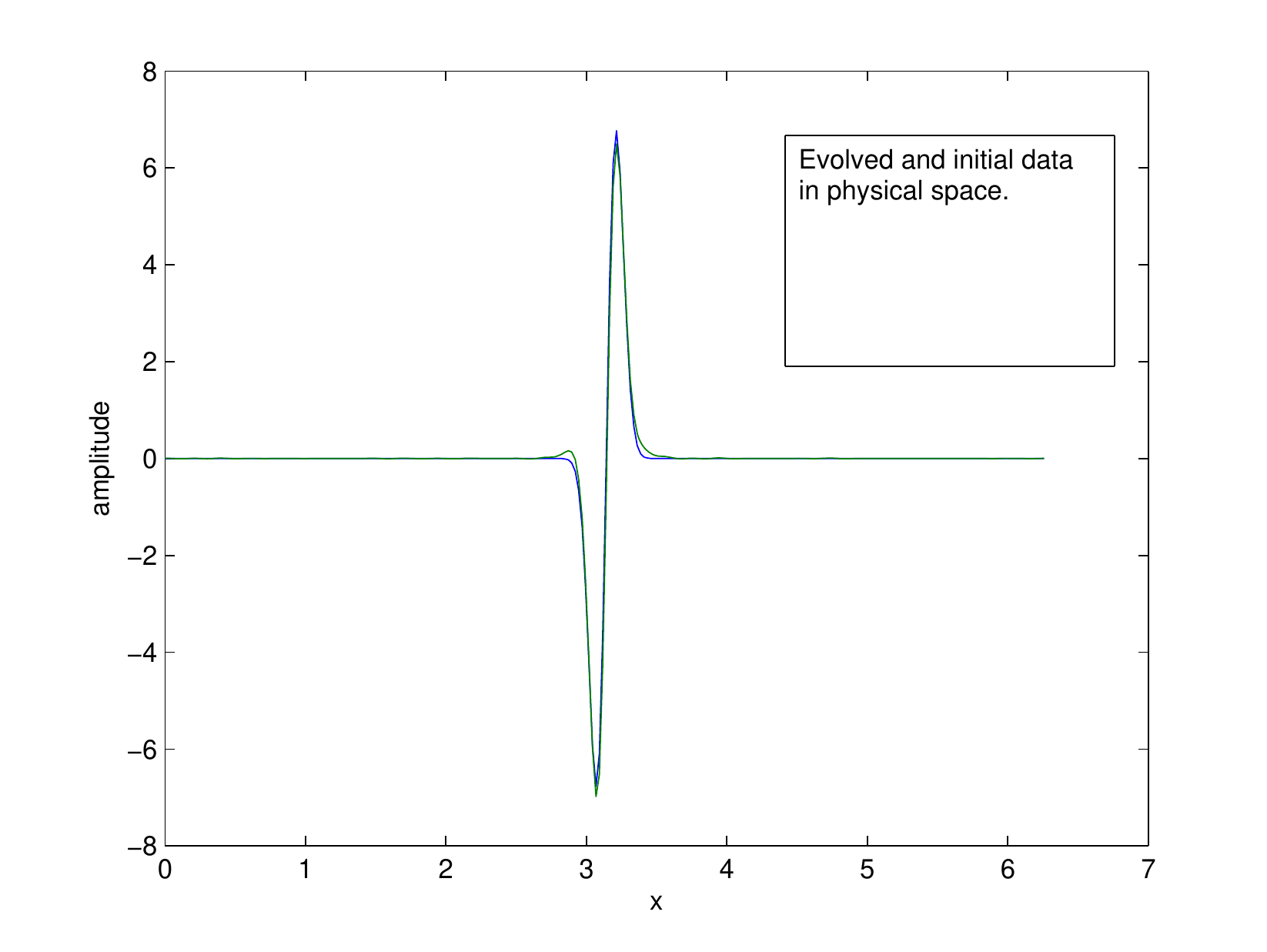}
\caption{The initial and evolved data. The time interval is $2\pi$, which is a period of linear KdV, therefore the initial and the evolved data are very close.}
\label{fig:phys2pi}
\end{figure}

The Figure \ref{fig:shorttime} demonstrates an obvious but important property that away from $t=2\pi N$, the evolved data is very far from the initial data.

\begin{figure}[htd]
\includegraphics[width=80mm]{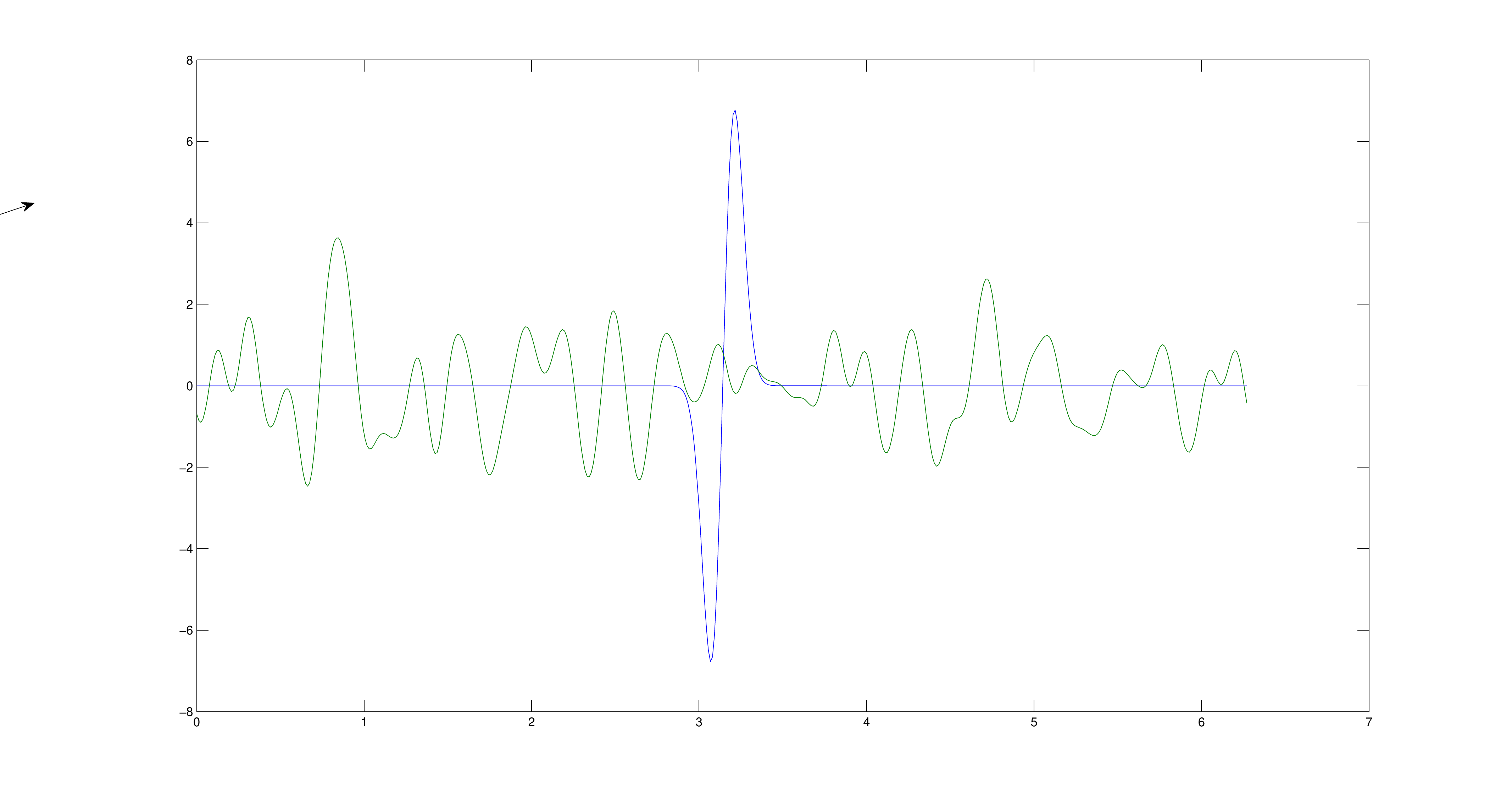}
\caption{Evolved data after short time. There is large distortion in physical space because of the strong dispersion. }
\label{fig:shorttime}
\end{figure}

\pagebreak

\section{Near-linear dynamics in water waves: asymptotics, 
numerical simulations and physical interpretation}
\label{sec:numerics}

Rogue waves are large-amplitude waves appearing on the sea surface seemingly ``from nowhere''.
Such abnormal waves have been also observed in shallow water and KdV has been used to explain this
phenomenon \cite{talipova}.
In the oceanographic literature, the following amplitude criterion for the rogue
 wave formation is generally used: the height of the rogue wave should
exceed the significant wave height by a factor of 2-2.2 \cite{KharifPelin}. (Significant wave height
 is the average wave height of the one-third largest waves.)

Major scenarios and explanations of rogue waves involve
\begin{itemize}
\item probabilistic approach: rogue waves are considered as rare events
in the framework of Rayleigh statistics
\item linear mechanism: dispersion enhancement (spatio-temporal focusing)
\item nonlinear mechanisms: in approximate models ({\em e.g.} KdV), for some special
initial data, large amplitude waves can be created.
\end{itemize}

Linear mechanism is very attractive as there are simple solutions leading to large amplitudes,
while nonlinear mechanism requires rather special initial data,
{\em e.g.} leading to the soliton formation. On the other hand, linear equations arise in
the small amplitude limit which is too restrictive.

Using near-linear dynamics in KdV one can experiment with another mechanism of large wave formation
that combines linear and nonlinear
deterministic mechanisms.  Our results indicate that for a special but relatively large set
of initial data  (characterized by  the energy contained mostly in high frequency Fourier modes),
the solutions of KdV equation behave near-linearly. It is then possible {\em to construct large
amplitude solutions using linear mechanisms of large wave formation}.

Here, we illustrate our approach with periodic boundary conditions. This model is not the most
realistic one but appropriate to illustrate the concept.

The KdV equation has been used to describe surface water waves in the small
amplitude limit of long waves in shallow water. More precisely, two  parameters
are assumed to be small and equal
\[
\frac{\rm amplitude}{\rm depth} \sim
\left ( \frac{\rm depth}{\rm wave length} \right )^2 \ll 1.
\]
Our numerical simulations of KdV show that the near-linear dynamics phenomenon occurs when small
parameter $\epsilon$  characterizing high frequency limit (see the formula below),
is only moderately small $\epsilon = 0.4$. The following three figures show that with $\epsilon = 0.4$,
 there is still a clear presence of near-linear evolution of KdV for some reasonable time interval $T=1$.
 
On the other hand, we will show that this value of $\epsilon=0.4$ is sufficiently large so that KdV still 
approximates shallow water waves dynamics.

\begin{figure}[b]
\begin{tabular}{cc}
\hspace{-15mm}\includegraphics[width=75mm]{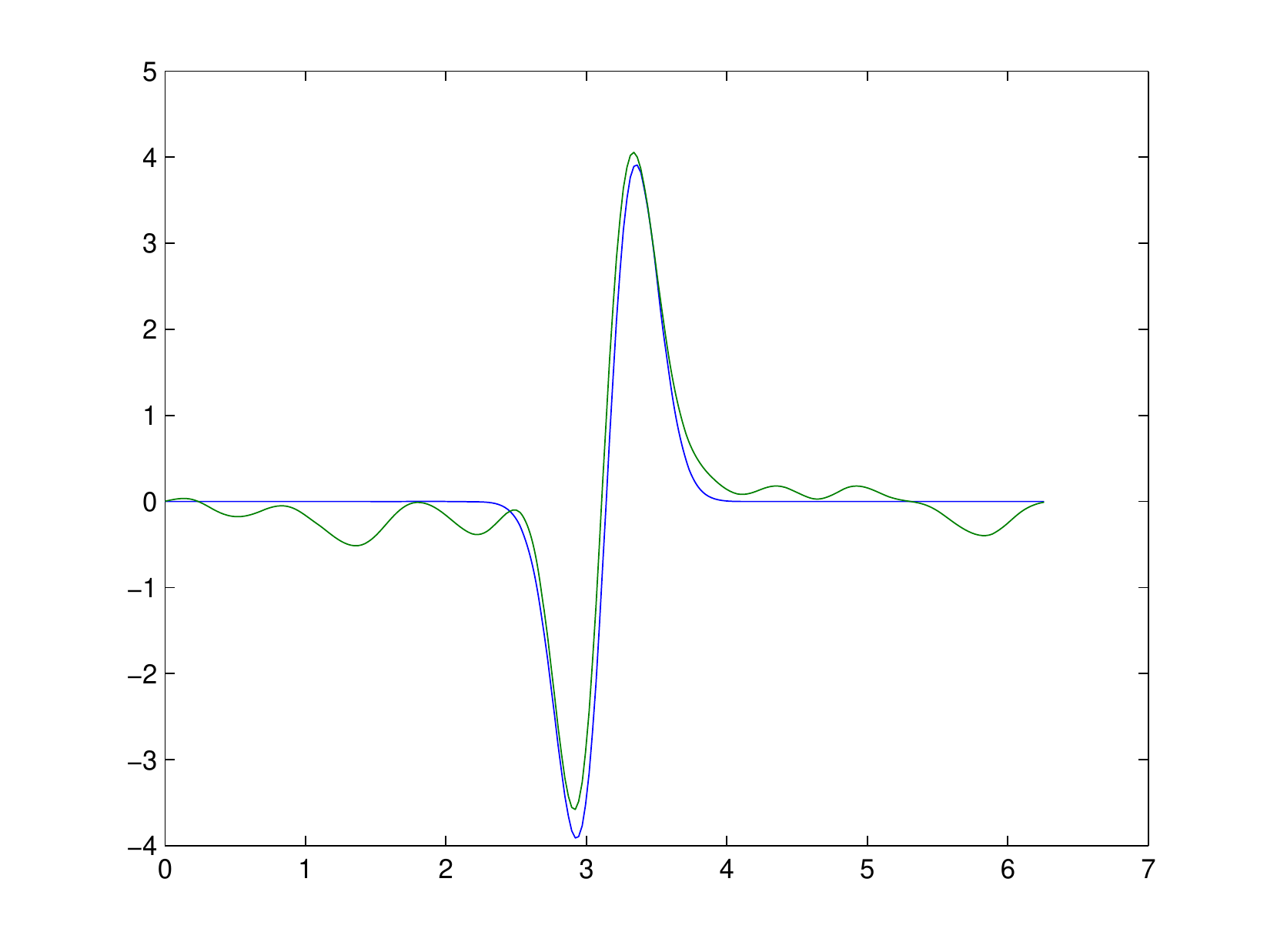} &
\includegraphics[width=75mm]{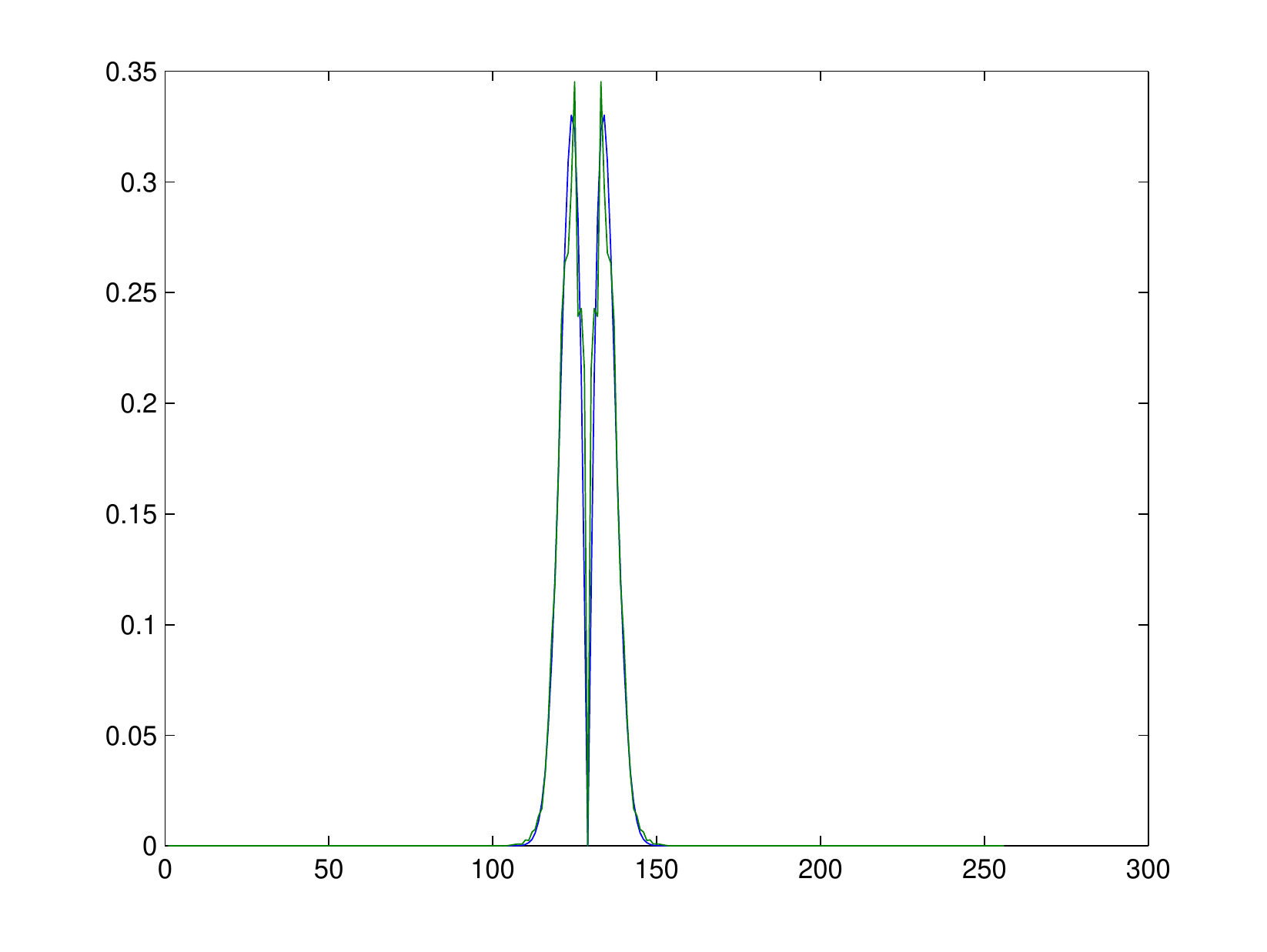}
\end{tabular}

\caption{Initial and evolved waves in KdV in Fourier space (right) and
 physical space (left).
The nonlinearly evolved data in physical space is pulled back
with reverse linear evolution, $e^{-Lt}u(x,t)$, for proper comparison.
 Abscissa shows the number of
Fourier harmonic (right) and spatial coordinate (left),
while ordinate is the amplitude. The time of evolution is $T =1$, {\bf $\epsilon=0.4$}.}
\label{fig1}
\end{figure}

\begin{figure}[htd]
\includegraphics[width =60mm]{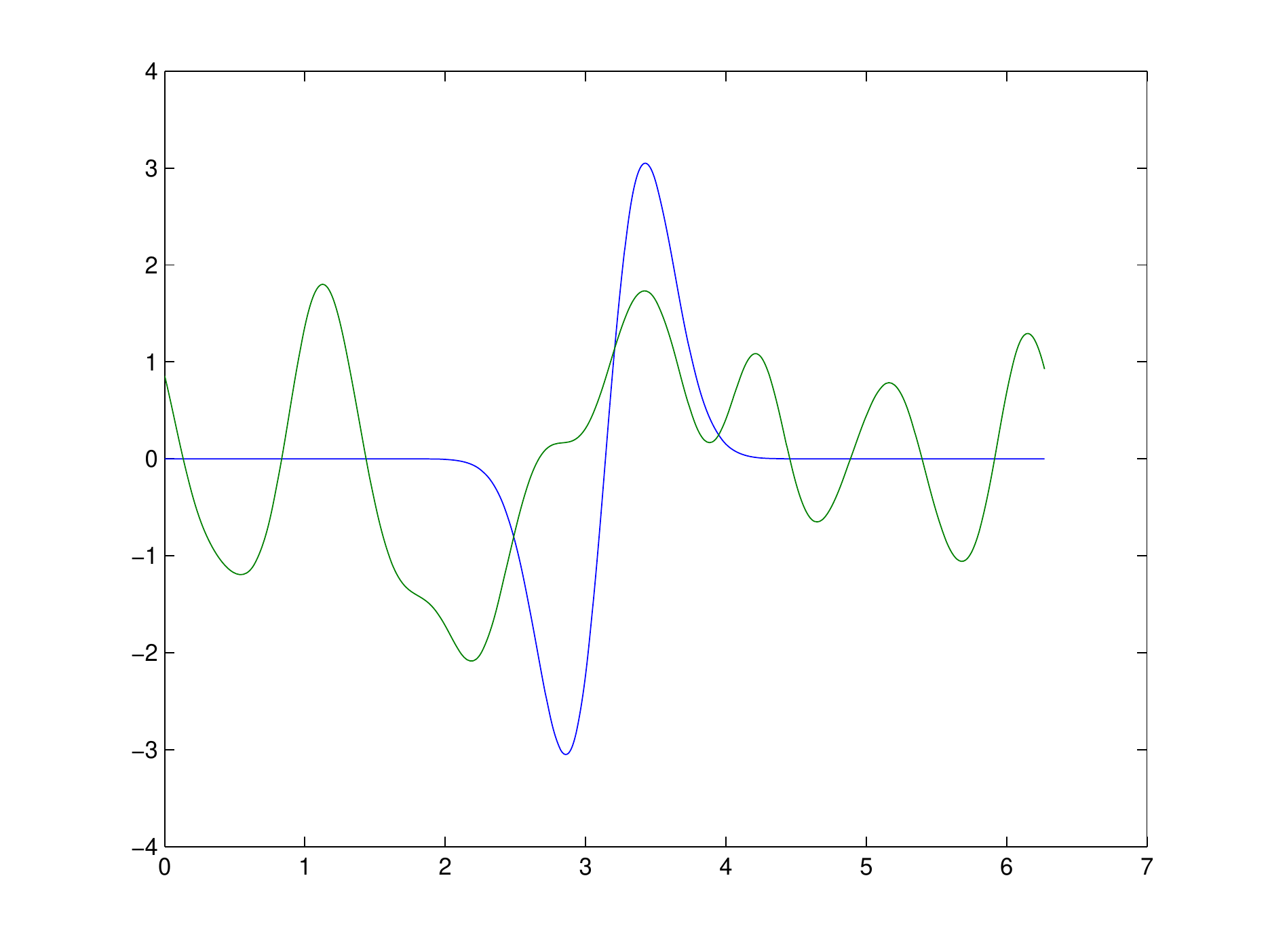}
\caption{The  initial data and the solution after $t=0.2$.
While, $\epsilon = 0.4$ is not so small, the dispersion is sufficiently strong so after a short 
time the initial wave disperses over the whole periodic domain.}
\label{fig2}
\end{figure}

As the initial data, we take the scaled 1st Hermit function
\[
u(x) = \frac{4.5}{\sqrt{\epsilon}}\left ( \frac{x}{\epsilon} \right ) e^{-\frac{x^2}{2\epsilon^2}},
\]
so that the energy $\int u^2 dx$ does not depend on $\epsilon$
and is very close to 1.
For numerical simulations, we use KdV in the form
\[
u_t = \frac{3}{2} u u_x + \frac{1}{6}u_{xxx}
\]
as it appears in the derivation of KdV from the water wave equations (see below).
Specific numerical parameters are:  the length of periodic domain $L=2\pi$. The number of modes $M=2^9$.
Time step size $\Delta t = 10^{-7}$ with the time of the evolution $T = 1$.
 The discretization in space is given by $h =L/M$.
We used the so-called Fornberg-Whitham scheme which is described in
\cite{fornberg}.

Now, using standard derivation of KdV from water waves equations, we recall the
 relation between physical parameters and rescaled dimensionless variables, see
\cite{whitham}, Chapter 13.11.

Let $h_0$ be the depth when the water is at rest and let $Y = h_0+\eta$ be
the free surface of the water. Let $a$ be a characteristic amplitude and $l$ be a
characteristic wave length. Assume that
\[
\alpha= \frac{a}{h_0} \sim \beta = \frac{h_0^2}{l^2} \ll 1.
\]
Both $\alpha$ and $\beta$ are small parameters in the problem and they must
be of the same order.

Next, use the following natural normalization
\[
x^{\prime} = l x, \,\,\, Y^{\prime} = h_0 Y, \,\,\, t^{\prime} = lt/c_0, \,\,\, \eta^{\prime} = a \eta,
\]
where primed variables are the original ones and $c_0=\sqrt{g h_0}$.

The formal asymptotic expansion leads to KdV with higher order corrections
\[
\eta_t + \eta_x +\frac{3}{2}\alpha \eta \eta_x + \frac{1}{6} \beta \eta_{xxx}+
O(\alpha^2 + \beta^2)=0.
\]
Let  $X = x-t$ and $T = \alpha t$, so the equation becomes
\begin{equation}
\eta_T +\frac{3}{2}\eta \eta_X + \frac{1}{6}  \eta_{XXX}+
O(\alpha + \beta^2/\alpha)=0.
\label{eq:lasteq}
\end{equation}
One should expect that this approximation has accuracy of the order $O(\alpha)$ for finite
time $T = O(1)$, which implies $t \sim \alpha^{-1}$ and
$t^{\prime} \sim l/(c_0\alpha)$.

Finally, since we modify our solution with another parameter $\epsilon$, we verify that KdV
approximation will still make sense for some choice of the parameters.

%\[
%\alpha= \frac{a}{h_0}, \,\,\, \beta = \frac{h_0^2}{l^2}, \,\,\, \alpha=\beta = \delta \ll 1.
%\]
First, let $\alpha=\beta = \delta \ll 1$.
Let us modify $a$ and $l$ with $ a_{\epsilon} = \frac{1}{\sqrt{\epsilon}} a$ and
 $l_{\epsilon} = \epsilon l$ which is consistent with our scaling of initial data. Then, we have

\[
\alpha_{\epsilon} = \frac{ a_{\epsilon}}{h_0} = \frac{\delta}{\sqrt{\epsilon}}, \,\,\,\,\,\,\,
\beta_{\epsilon} = \frac{h_0^2}{l_{\epsilon}^2} = \frac{\delta}{\epsilon^2}.
\]
These are small with $\delta = 0.01$ and $\epsilon = 0.4$. On the other hand the "mismatch" in the
equation (\ref{eq:lasteq}) is
\[
\alpha_{\epsilon}+ \frac{\beta_{\epsilon}^2}{\alpha_{\epsilon}} = \frac{\delta}{\sqrt{\epsilon}}+
\frac{\delta}{\epsilon^{3.5}} \approx \frac{1}{4}.
\]
Therefore, our high frequency regime may approximate water waves dynamics for example
with the following parameters: $a  =  1$ m, $h_0  =  100$ m, and
$l  =  1000$ m.

\end{document}